\documentclass[pra,aps,floatfix,superscriptaddress,notitlepage,nofootinbib]{revtex4-1}
\usepackage[utf8]{inputenc}

\usepackage{fullpage} 
\usepackage{tikz} 
\usepackage{amsmath,amsfonts,amsthm, mathtools, mathrsfs}
\usepackage{float}
\usepackage{dsfont}
\usepackage{csquotes}
\usepackage{amssymb}
\usepackage{verbatim}
\usepackage{enumerate}
\usepackage[toc,page]{appendix}
\usepackage{centernot}

\usepackage{framed}

\usepackage{graphicx}{\small }
\usepackage[leftcaption]{sidecap}

\usepackage{hyperref}
\definecolor{darkred}  {rgb}{0.5,0,0}
\definecolor{darkblue} {rgb}{0,0,0.5}
\definecolor{darkgreen}{rgb}{0,0.5,0}
\hypersetup{
	colorlinks = true,
	urlcolor  = blue,         
	linkcolor = darkblue,     
	citecolor = darkgreen,    
	filecolor = darkred       
}

\theoremstyle{theorem}
\newtheorem{corollary}{Corollary}

\newtheorem{proposition}{Proposition}
\newtheorem{theorem}{Theorem}

\theoremstyle{definition}
\newtheorem{definition}{Definition}

\newtheorem*{remark}{Remark}

\newcommand{\N}{\mathcal{N}}

\definecolor{cool_green}{rgb}{0.0, 0.5, 0.0}
\definecolor{cool_blue}{rgb}{0.0, 0.0, 0.5}

\newcommand{\Tr}[1]{\operatorname{Tr}\left\{#1\right\}}

\def\>{\rangle}
\def\<{\langle}
\def\N#1{\left|\!\left|{#1}\right|\!\right|}

\def\mE{\mathcal{E}}

\def\openone{\mathds{1}}
\newcommand{\set}[1]{\mathsf{#1}}

\newcommand{\Cm}{\mathbb{C}^m}
\newcommand{\Cn}{\mathbb{C}^n}
\newcommand{\Cd}{\mathbb{C}^d}
\newcommand{\nR}{\mathfrak{F}}

\makeatletter
\newcommand*\bigcdot{\mathpalette\bigcdot@{.5}}
\newcommand*\bigcdot@[2]{\mathbin{\vcenter{\hbox{\scalebox{#2}{$\m@th#1\bullet$}}}}}
\makeatother

\renewcommand{\qedsymbol}{\nobreak \ifvmode \relax \else
	\ifdim \lastskip<1.5em \hskip-\lastskip \hskip1.5em plus0em
	minus0.5em \fi \nobreak \vrule height0.75em width0.5em
	depth0.25em\fi}

\renewcommand{\ge}{\geqslant}
\renewcommand{\le}{\leqslant}

\renewcommand{\succeq}{\succ}

\begin{document}
	
\title{General state transitions with exact resource morphisms:\\a unified resource-theoretic approach}

\author{Wenbin Zhou}
\email{zhou.wenbin@i.mbox.nagoya-u.ac.jp}
\affiliation{Graduate  School  of  Informatics,  Nagoya
	University, Chikusa-ku, 464-8601 Nagoya, Japan}

\author{Francesco Buscemi}
\email{buscemi@i.nagoya-u.ac.jp}
\affiliation{Graduate  School  of  Informatics,  Nagoya
	University, Chikusa-ku, 464-8601 Nagoya, Japan}

\date{\today}	

\begin{abstract}
	
	Given a non-empty closed convex subset $\set{F}$ of density matrices, we formulate conditions that guarantee the existence of an $\set{F}$-morphism (namely, a completely positive trace-preserving linear map that maps $\set{F}$ into itself) between two arbitrarily chosen density matrices. While we allow errors in the transition, the corresponding map is required to be an exact $\set{F}$-morphism. Our findings, though purely geometrical, are formulated in a resource-theoretic language and provide a common framework that comprises various resource theories, including the resource theories of bipartite and multipartite entanglement, coherence, athermality, and asymmetric distinguishability.

	We show how, when specialized to some situations of physical interest, our general results are able to unify and extend previous analyses. We also study conditions for the existence of maximally resourceful states, defined here as density matrices from which any other one can be obtained by means of a suitable $\set{F}$-morphism. Moreover, we quantitatively characterize the paradigmatic tasks of optimal resource dilution and distillation, as special transitions in which one of the two endpoints is maximally resourceful.
\end{abstract}

\maketitle

\section{Introduction}

Many problems in information theory and statistics (and quantum generalizations thereof) can be reformulated as to whether a suitable transition between two objects is possible under suitable constraints. For example, noisy channel coding in information theory can be understood as the problem of transforming a number of uses of a noisy channel into less uses of a less noisy (ideally noiseless) channel, under the constraint that no side communication is available between sender and receiver. Another example is provided in mathematical statistics by the task of extracting statistics from a given sample. The exact same intuition becomes paradigmatic in statistical mechanics, most notably in thermodynamics, where a typical problem, for example, is to whether a thermodynamic process satisfying certain constraints (e.g., an isothermal process, an adiabatic process, etc.) exists between two given states.

Whenever the set of allowed transformations is constrained for some reasons, it is natural to consider the possibility of performing a forbidden transformation as a \textit{resource}~\cite{Devetak-HArrow-Winter-resource,horodecki09,COECKE201659,Chitambar2019}. The state space of the system (that is, the set of objects that are being transformed) thus ends up being divided into free states (that is, states that can be reached from \textit{any other} state by means of allowed transformations) and non-free states. The latter are implicitly defined to be the \textit{resourceful} states in the theory. Notice that the word ``state'' here is not to be meant in its strict sense of ``state of a physical system'', but denotes more generally the objects that are being transformed. These could comprise states proper, like in conventional thermodynamics, but are not limited to these. One could consider as objects, for example, statistical models (as in mathematical statistics), or noisy channels (as in information theory), or entire generalized operational theories for that matter.

In fact, it turns out that it is often easier to study resource theories, in which the first thing being defined is not the set of allowed transformations, but rather the set of free states, while allowed transformations are \textit{implicitly} defined only later, as all those transformations that map the free set into itself. Even though sometimes less clear operationally, such an approach (that we may call ``geometric'', in order to contrast it with the previous ``operational'' one that puts constraints on allowed operations) has, from a purely mathematical viewpoint, various advantages: for example, resource theories typically become asymptotically reversible only if the set of allowed transformations is taken to be the maximal one (plus \textit{epsilon}) compatible with the theory \cite{Brandao--Plenio,Brandao--Gour}.

In this paper we study general resource theories by following the geometric approach, that is, by building the theory on top of a given set of free states $\set{F}$, that we in particular assume to be non-empty, closed, and convex. Accordingly, we allow any transformation that maps $\set{F}$ into itself. For the sake of generality, in this paper we do not even assume any particular rule of composition for free states, so that we essentially work in the single-shot regime, but we still allow the input and the output systems, and their respective free sets, to change as a consequence of the transformation. In this way, even though we cannot say anything about asymptotic rates, we can still discuss in full generality single-shot rates for resource distillation and dilution, for example.

The main aspect that differentiates our work from previous ones is that here we do not focus only on the tasks of resource dilution and distillation. Instead, we consider the more general problem of formulating sufficient (and, in some cases, necessary) conditions for the existence of an $\set{F}$-preserving transformation between an \textit{arbitrarily} given input--target pair of states. This means that the analysis presented here does not rely on the existence of any privileged maximally resourceful state (like the maximally entangled state in bipartite entanglement theory) and thus applies to quite general resource theories.

Concerning the kind of transformations that we consider in this work, two remarks are in order. First, while we allow for noisy transitions, that is, situations in which the input state is mapped not exactly but ``close'' to the wanted target, we only allow transitions implemented by operations that are \textit{exactly} $\set{F}$-preserving. This stands in contrast to other single-shot analyses of general resource theories in which ``almost-free'' operations are allowed in the finite block-length regime (see, e.g., Ref.~\cite{Brandao--Gour}). Second, we search for conditions that are expressed in terms of resource monotones (a notion to be rigorously introduced in what follows), computed \textit{separately} for the initial state and for the target state. On the one hand, in this way we are able to ``decouple'' the initial state from the target state, and to speak of their respective resource's worth, independently of each other. On the other hand, this means that functions that comprise both states at once, like those that in some situations can be obtained by means of semidefinite linear programs~\cite{Buscemi2017,Gour-Jennings-Buscemi-Duan}, will not be considered in this work. Another aspect of the problem that is not considered in this work is the computational complexity of computing resource monotones.

The paper is structured as follows. After introducing the notation and reviewing some basic notions in Section~\ref{sec:preliminaries}, we present our main results and prove them in detail in Section~\ref{sec:main}. Section~\ref{sec:applications} presents some applications of our general analysis to specific cases of physical interest: applications to the resource theories of athermality and asymmetric distinguishability (Section~\ref{sec:singleton}); bipartite entanglement (Section~\ref{sec:entanglement}); conditions for the existence of maximally resourceful states, together with necessary and sufficient conditions for dilution and distillation which can give, whenever a dimension scaling is provided, optimal dilution and distillation rates (Section~\ref{sec:maxres}). Section~\ref{conclusion} concludes the paper.

\section{Mathematical Preliminaries}\label{sec:preliminaries}

We denote by $\set{D}(\Cm)$ the set of all $m$-by-$m$ complex density matrices $\rho$, i.e., $\rho\ge 0$ and $\Tr{\rho}=1$, which are used here to represent quantum states of $m$-dimensional quantum systems. Within $\set{D}(\Cm)$, we identify a non-empty closed convex subset $\set{F}$ as the set of ``free states''. The closure and the convexity of $\set{F}$ is crucial in various steps of our proofs, for example, when invoking the closure under convex mixtures, or when applying a variant of the minimax theorem that requires convex domain. Here and throughout this work, \textit{resource morphisms} (or more precisely $\set{F}$-morphisms\footnote{Here we prefer the term ``resource morphisms'' to the more common ``free operations'' because it reminds the fact that the foundational concept, in the geometric approach in which we are working, is the free set $\set{F}$, not the set of allowed transformations, which are just defined as all those that map $\set{F}$ into itself.}) are defined as completely positive, trace-preserving (CPTP) linear maps $\mE:\set{D}(\Cm)\to \set{D}(\Cm)$ such that $\mE(\set{F})\subseteq \set{F}$. More generally, one may consider CPTP maps that change the dimension of the system, for example, $\mE:\set{D}(\Cm)\to \set{D}(\Cn)$. Also in this case, whenever the \textit{output} free set $\set{F}'$ is also specified, it is possible to define a notion of resource morphisms by the condition that $\mE(\set{F})\subseteq \set{F}'$. However in what follows, for the sake of readability, we will restrict ourselves to the case of equal input and output dimensions and $\set{F}=\set{F}'$, keeping in mind however that all the results we derive can be straightforwardly extended to the general case. We will go back to the more general setting, with different input and output systems, in Section~\ref{sec:applications} when discussing various applications like the tasks of resource dilution and distillation.

The general structure that we study in this work is the following:

\begin{framed}
	
\begin{definition}[Resourcefulness Preorder]\label{def:preorder}
	Given two density matrices $\rho, \sigma \in \set{D}(\Cm)$, we write $	\rho \succeq_\epsilon \sigma$ whenever there exists a resource morphism $\mE$ such that 
	$\frac12\N{\sigma-\mE(\rho)}_1\le \epsilon$, where $\N{X}_1:=\Tr{\sqrt{X^\dag X}}$ denotes the trace-norm. In particular, we write $\rho\succeq\sigma$ whenever $\rho\succeq_{\epsilon=0}\sigma$. 
\end{definition}

\end{framed}

In the above definition, the preorder $\succeq_{\epsilon}$ has been introduced with respect to the distance induced by the trace-norm, although it is possible to use any other well-behaved distance measure between density matrices (like the fidelity, for example) without substantially changing the results~\cite{Nielsen--Chuang,wildebook13}. In any case, an important thing to notice in Definition~\ref{def:preorder} is that, even though errors are allowed in the state transformation, we always require the constraint $\mE(\set{F})\subseteq\set{F}$ to be \textit{strictly} satisfied.

The resourcefulness preorder naturally lead us to define a maximally resourceful element as follows:

\begin{framed}
\begin{definition}[Maximally resourceful element]\label{def:max-res-state}
	An element $\alpha \in \set{D}(\Cm)$ is said to be \textit{maximally resourceful} if $\alpha\succ\sigma$ for any $\sigma \in \set{D}(\Cm)$.
\end{definition}
\end{framed}

Given a general resource theory, an important question is to whether the theory possesses maximally resourceful elements or not. In Section~\ref{sec:maxres} we will consider sufficient conditions for their existence. However, we recall that our main results do not rely in any way on the existence of maximally resourceful elements.

\subsection{Information-theoretic divergences}

In what follows, for any operator $\rho\in\set{D}(\Cm)$ we denote by $\Pi_\rho$ the orthogonal projector onto its support (i.e., the orthogonal complement of its kernel). Moreover, for any $\epsilon\in[0,1]$, we denote by $\mathsf{B}^\epsilon(\rho)$ the set of operators $\{\rho'\in\set{D}(\Cm):\N{\rho-\rho'}_1\le2\epsilon \}$ and by $\mathsf{P}^\epsilon(\rho)$ the set of operators $\{P:0\le P\le\openone\text{ and }\Tr{\rho P}\ge 1-\epsilon \}$. All logarithms are taken in base 2.

\begin{framed}
\begin{definition}[Relative entropies]
	Given two density matrices $\rho,\sigma\in\set{D}(\Cm)$, we define
	\begin{enumerate}[(i)]
		\item the \textit{Umegaki relative entropy}~\cite{umegaki62}:
		\begin{align}
		D(\rho\|\sigma):=
		\begin{cases}
		\Tr{\rho\ (\log\rho-\log\sigma)}\;,&\text{ if }\Pi_\sigma\ge\Pi_\rho\;,\\
		+\infty\;,&\text{ otherwise}\;;
		\end{cases}
		\end{align}
		\item the \textit{hypothesis testing relative entropy}~\cite{datta09}: for any $\epsilon\in[0,1]$
		\begin{align}
		D_h^{\epsilon}(\rho\|\sigma):=-\log\quad \min_{P\in\mathsf{P}^\epsilon(\rho)}\quad\Tr{\sigma\ P} \;,
		\end{align}
		with the convention $-\log 0:=+\infty$; for $\epsilon=0$, one recovers the \textit{min-divergence}, defined as \cite{datta08}
		\begin{equation}
		D_{\min}(\rho\|\sigma):=-\log \Tr{\sigma\ \Pi_{\rho}}\;;
		\end{equation}
		\item the \textit{max-divergence} \cite{datta08}:
		\begin{equation}
		D_{\max}(\rho\|\sigma):=
		\begin{cases}
		\log\min\{\lambda\in\mathbb{R}:\lambda\sigma-\rho\ge 0 \}\;,&\text{ if }\Pi_\sigma\ge\Pi_\rho\;,\\
		+\infty\;,&\text{ otherwise}\;;
		\end{cases}
		\end{equation}
		in this case we also define a ``smoothed'' version as follows: for any $\epsilon\in[0,1]$, 
		\begin{equation}\label{eq:dmax-smooth}
		D_{\max}^{\epsilon}(\rho\|\sigma):= \inf_{\rho'\in\mathsf{B}^{\epsilon}(\rho)}D_{\max}(\rho'\|\sigma)\;.
		\end{equation}
		\end{enumerate}
\end{definition}
\end{framed}

A crucial property satisfied by all these divergences is the monotonicity under CPTP linear maps, that is, for example, $D(\mE(\rho)\|\mE(\sigma)\le D(\rho\|\sigma)$ and analogously for the others. In a resource theory characterized by the set of free states $\set{F}$, we also introduce the following:

\begin{framed}
\begin{definition}[Max-divergence of resourcefulness]\label{def:max-div-resourcefulness}
	Given two density matrices $\rho,\sigma\in\set{D}(\Cm)$ and a non-empty closed convex subset $\set{F}\subseteq\set{D}(\Cm)$, the max-divergence relative to $\set{F}$ is defined as
	\begin{equation}
	D_{\max,\set{F}}(\rho\|\sigma):=
	\log\inf\left\{\lambda\in\mathbb{R}:\frac{\lambda\sigma-\rho}{\lambda-1}\in\set{F} \right\}\;,
	\end{equation}
	with the convention that $\inf\varnothing=+\infty$. Its ``smoothed'' version is defined in analogy with~\eqref{eq:dmax-smooth}, that is
	\begin{equation}
	D_{\max,\set{F}}^\epsilon(\rho\|\sigma):=\inf_{\rho'\in\mathsf{B}^{\epsilon}(\rho)}D_{\max,\set{F}}(\rho'\|\sigma)\;.
	\end{equation}
\end{definition}
\end{framed}

We notice that if $\set{F}=\set{D}(\Cm)$, then $D_{\max,\set{F}}(\rho\|\sigma)=D_{\max}(\rho\|\sigma)$, but in general $D_{\max,\set{F}}(\rho\|\sigma)\ge D_{\max}(\rho\|\sigma)$. Moreover, while $D_{\max,\set{F}}$ may fail to be monotonic under general CPTP maps, it is monotonic under the action of resource morphisms, that is, CPTP maps that map $\set{F}$ into itself. This can be easily seen by noticing that, if for some $\lambda$, $(\lambda-1)^{-1}(\lambda\sigma-\rho)$ is in $\set{F}$, then, for any resource morphism $\mE$, also $(\lambda-1)^{-1}\mE(\lambda\sigma-\rho)=(\lambda-1)^{-1}(\lambda\mE(\sigma)-\mE(\rho))$ is automatically in $\set{F}$, so that $D_{\max,\set{F}}(\mE(\rho)\|\mE(\sigma) )$ cannot be larger than $D_{\max,\set{F}}(\rho\|\sigma)$.

\subsection{Resource monotones}

We say that a function $f:\set{D}(\Cm)\to[0,+\infty]$ constitutes a \textit{resource monotone} if it achieves its global minimum on all elements of $\set{F}$, and it does not increase under the action of resource morphisms, i.e., $f(\rho)\ge f(\mE(\rho))$ for any resource morphism $\mE$. More properties can be demanded (and are indeed desirable) in order to fruitfully work with concrete examples of resource monotones. The information-theoretic divergences introduced above can be used to introduce resource monotones that inherit many useful properties from the parent divergence. In our construction, the following quantities play a central role~\cite{Liu-Bu-Takagi-PRL}.

\begin{framed}
\begin{definition}[Entropic Resource Monotones]\label{def:entropic-monotones}
	Given a non-empty closed convex set $\set{F}\subseteq\set{D}(\Cm)$, for any density matrix $\rho\in\set{D}(\Cm)$ and any $\epsilon\in[0,1]$, we define the following quantities:
	\begin{enumerate}[(i)]
		\item $\mathfrak{D}(\rho):=\inf_{\omega \in \set{F}}D(\rho\|\omega)$;
		\item $\mathfrak{D}_h^\epsilon(\rho):=-\log \max_{\omega \in \set{F}}\min_{P\in\mathsf{P}^\epsilon(\rho)}\Tr{P\ \omega}$, with the convention $-\log 0:=+\infty$;
		\item $\mathfrak{D}_{\max}^\epsilon(\rho):=\inf_{\omega \in \set{F}}D_{\max}^\epsilon(\rho\|\omega)$;
		\item $\mathfrak{D}_{\max,\set{F}}^\epsilon(\rho):=\inf_{\omega \in \set{F}}D^\epsilon_{\max,\set{F}}(\rho\|\omega)$.
	\end{enumerate}
	In the case $\epsilon=0$, we simply remove the superscript; the only exception is $\mathfrak{D}_h^{\epsilon=0}(\rho)$, for which we will use the special notation $\mathfrak{D}_{\min}(\rho)$.
\end{definition}
\end{framed}

The above quantities are all well-behaved resource monotones. This fact is a direct consequence of the monotonicity of the parent divergences under the action of resource morphisms.

\begin{framed}
\begin{definition}[Free fraction and generalized free fraction]\label{def:robustnesses}
	Given a non-empty closed convex free set $\set{F}\subseteq\set{D}(\Cm)$, the \textit{free fraction} of a density matrix $\rho\in\set{D}(\Cm)$ is defined by the formula
	\begin{equation}\label{eq:robustness}
	\nR(\rho):=\max\{p\in[0,1]: \exists\ \omega \in \set{F} 	\text{ s.t. }p\rho+(1-p)\omega\in \set{F} \}\;.
	\end{equation}
	When mixing with general $\omega\in\set{D}(\Cm)$ instead of $\omega\in\set{F}$ is allowed, one obtains the \textit{generalized free fraction}, defined as
	\begin{equation}\label{eq:gen-robustness}
	\nR_g(\rho):=\max\{p\in[0,1]: \exists\ \omega \in \set{D}(\Cm) 	\text{ s.t. }p\rho+(1-p)\omega\in \set{F} \}\;.
	\end{equation}
\end{definition}
\end{framed}

The free fraction and the generalized free fraction are related to the \textit{robustness} $\mathfrak{R}(\rho)$~\cite{Vidal--Tarrach--Rolf} and the \textit{generalized robustness} $\mathfrak{R}_g(\rho)$~\cite{Steiner--Michael}, respectively, through the relations $\nR(\rho)^{-1}=1+\mathfrak{R}(\rho)$ and $\nR_{g}(\rho)^{-1}=1+\mathfrak{R}_g(\rho)$, and they are both directly related with the entropic resource monotones in Definition~\ref{def:entropic-monotones} as follows:
	\begin{align}
	&-\log{\nR(\rho)}=\log(1+\mathfrak{R}(\rho))=\mathfrak{D}_{\max,\set{F}}(\rho)\;,\\
	&-\log{\nR_{g}(\rho)}=\log(1+\mathfrak{R}_g(\rho))=\mathfrak{D}_{\max}(\rho)\;,
	\end{align}
with the convention $-\log 0:=+\infty$. In particular, we have that $\mathfrak{D}_{\max}(\rho)$ coincides with the \textit{generalized logarithmic robustness} of \cite{datta08,datta-LogRobustness}, while $\mathfrak{D}_{\max,\set{F}}(\rho)$ coincides with the \textit{logarithmic robustness} of~\cite{brandao--datta}. For this reason, in what follows, when speaking of $\mathfrak{D}_{\max}(\rho)$ (respectively, $\mathfrak{D}_{\max,\set{F}}(\rho)$) we will follow the mainstream convention and call it ``generalized log-robustness'' (respectively, ``log-robustness'') even though, depending on the context, it would be more appropriate to use the term we introduced above, that is, ``generalized log-free fraction'' (respectively, ``log-free fraction'').

\begin{remark}
	All resource monotones introduced above would still be well-defined monotones even if the class of resource morphisms were enlarged to comprise also positive, but not completely-positive, linear maps. This seems no coincidence, since at the single-shot level, where no rule for composing system is given yet, there really is no compelling mathematical reason to limit the discussion to CPTP linear maps only. This is a common feature of various problems in quantum statistics, in particular quantum decision theory, where the theory becomes simpler if one works with quantum statistical morphism (which may violate complete positivity) and introduce CPTP maps as special cases, rather than starting from the beginning with fully blown CPTP maps~\cite{Buscemi2012}. Here we do not investigate further into this point, and simply justify the assumption of CPTP-ness on practical grounds.
\end{remark}

\subsection{Optimal convex decompositions}\label{sec:convex-dec}

Our main results rely on the following construction, whose intuitive picture is given in Fig.~\ref{fig:generalrobustness} below.

\begin{figure}[h!]
	\includegraphics[scale=0.3]{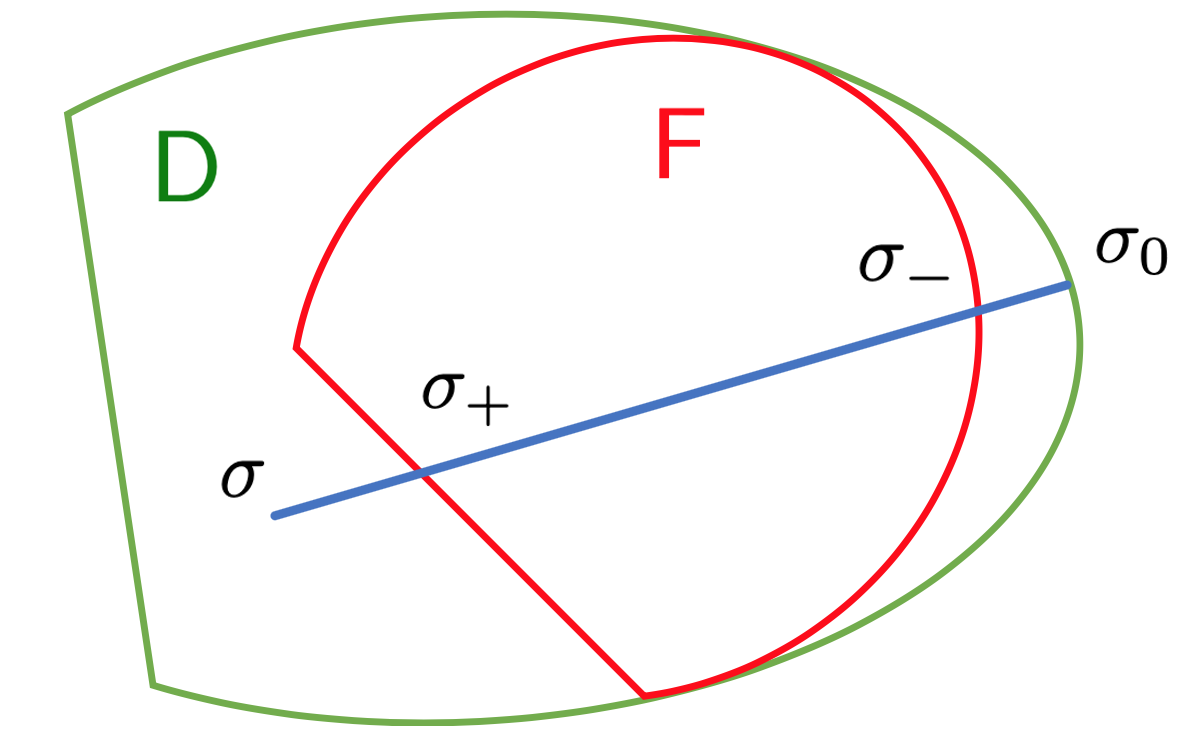}
	\caption{Geometric intuition for the generalized free fraction introduced in Definition~\ref{def:robustnesses}. Here $\sigma_{0}$ denotes the optimized density matrix, which is able to achieve, by means of convex mixing, the generalized free fraction of $\sigma$.}
	\label{fig:generalrobustness}
\end{figure}

Given $\sigma\in\set{D}(\Cm)$, assuming $\sigma\notin\set{F}$, let us fix a convex decomposition achieving its generalized free fraction and write it as
\begin{align}\label{eq:fraction}
\sigma_{+}=\nR_{g}(\sigma) \sigma+[1-\nR_{g}(\sigma)]\sigma_0\;.
\end{align}
In the above equation, due to the optimality of $\nR_{g}$, $\sigma_+\in\set{F}$ lies on the border of $\set{F}$, while $\sigma_0$ lies on the border of $\set{D}(\Cm)$, as depicted in Fig.~\ref{fig:generalrobustness}. The above decomposition includes the situation in which $\nR_g(\sigma)=0$, that is, $\sigma_+=\sigma_0$. For any decomposition as in~\eqref{eq:fraction}, another free state $\sigma_{-}$ can be uniquely defined using the max-divergence of resourcefulness (Definition~\ref{def:max-div-resourcefulness}) as follows:
\begin{align}
	\label{robustness ratio}
	\sigma_{-}:&=\frac{2^{D_{\max,\set{F}}(\sigma\|\sigma_{+})}\sigma_{+}-\sigma}{2^{D_{\max,\set{F}}(\sigma\|\sigma_{+})}-1}\\
	&=\left[\frac{\nR_{g}(\sigma)\,2^{D_{\max,\set{F}}(\sigma\|\sigma_{+})}-1}{2^{D_{\max,\set{F}}(\sigma\|\sigma_{+})}-1}\right]\sigma+\left[1-\frac{\nR_{g}(\sigma )\,2^{D_{\max,\set{F}}(\sigma\|\sigma_{+})}-1}{2^{D_{\max,\set{F}}(\sigma\|\sigma_{+})}-1}\right]\sigma_0\;,\label{eq:to-derive}
\end{align}
whenever $D_{\max,\set{F}}(\sigma\|\sigma_{+})<+\infty$, or $\sigma_{-}:=\sigma_{+}$ otherwise. In order to derive~\eqref{eq:to-derive} we just plugged~\eqref{eq:fraction} into~\eqref{robustness ratio} and rearranged terms. Notice that since $\sigma\notin\set{F}$, we have $\sigma\neq\sigma_+$ and $D_{\max,\set{F}}(\sigma\|\sigma_{+})>0$. It is easy to check that, by construction, $\sigma_{-}$, as $\sigma_+$, lies on the intersection between the border of $\set{F}$ and the segment joining $\sigma$ with $\sigma_0$. Our main results will originate from a careful evaluation of the relative distances between these four points in state space.

\section{Main Results}\label{sec:main}

In this section, we state and prove the main results of this paper. Firstly we derive, for any finite-dimensional resource theory in which the set of free states is non-empty closed and convex, sufficient conditions for the existence of a resource morphism between any two states, given in terms of resource monotones. Such conditions are formulated so to allow, in general, non-zero errors in the state transition, while the operation implementing the transition is an exact resource morphism.

\begin{framed}
\begin{theorem}\label{th:3}
	Let us arbitrarily fix two states, $\rho,\sigma\in\set{D}(\Cm)$, and two values $\epsilon_1, \epsilon_2 \in[0,1]$. Let us moreover choose $\tilde{\sigma}\in\mathsf{B}^{\epsilon_2}(\sigma)$ and $\tilde\sigma_{+}\in\set{F}$ so that $D_{\max}(\tilde\sigma\|\tilde\sigma_+) =\mathfrak{D}^{\epsilon_2}_{\max}(\sigma)$.
	 \begin{enumerate}[(i)]
	 	\item If $\mathfrak{D}^{\epsilon_1}_{h}(\rho)=+\infty$, then $\rho\succeq_{\epsilon_1}\sigma$.
	 	\item If $\mathfrak{D}^{\epsilon_2}_{\max}(\sigma)=0$, then $\rho\succeq_{\epsilon_2}\sigma$.
	 	\item If $\mathfrak{D}^{\epsilon_1}_{h}(\rho)<+\infty$ and $\mathfrak{D}^{\epsilon_2}_{\max}(\sigma)>0$, then
	 	\begin{enumerate}[(a)]
	 		\item either $D_{\max,\set{F}}(\tilde{\sigma}\|\tilde\sigma_{+})<+\infty$; in such a case, $\rho \succeq_{\epsilon_1+\epsilon_2} \sigma$ if the following two conditions simultaneously hold:
	 		\begin{align}\label{error first}
	 		\mathfrak{D}^{\epsilon_1}_{h}(\rho)	\ge		 \mathfrak{D}^{\epsilon_2}_{\max}(\sigma)
	 		\end{align}
	 		and
	 		\begin{align}
	 		\label{error second}
	 		& 2^{-\max_{\omega \in \set{F}}D^{\epsilon_1}_{h}(\rho\|\omega)}
	 		\ge \frac{2^{D_{\max,\set
	 					F}(\tilde{\sigma}\|\tilde\sigma_{+})-\mathfrak{D}^{\epsilon_1}_{h}(\rho)}-1}{2^{D_{\max,\set
	 					F}(\tilde{\sigma}\|\tilde\sigma_{+})}-1} \;;
	 		\end{align}
	 		\item or $D_{\max,\set{F}}(\tilde{\sigma}\|\tilde\sigma_{+})=+\infty$; in such a case, $\rho \succeq_{\epsilon_1+\epsilon_2} \sigma$ if condition~\eqref{error first} above holds together with
	 		\begin{align}\label{eq:max=min}
	 		\max_{\omega \in \set{F}}D^{\epsilon_1}_{h}(\rho\|\omega)=\min_{\omega \in \set{F}}D^{\epsilon_1}_{h}(\rho\|\omega)\;.
	 		\end{align}
	 	\end{enumerate} 
\end{enumerate}
\end{theorem}
\end{framed}

\begin{remark}
	As discussed in Section~\ref{sec:convex-dec}, the assumption $\mathfrak{D}^{\epsilon_2}_{\max}(\sigma)>0$ in case~(iii.a) guarantees that also $D_{\max,\set{F}}(\tilde{\sigma}\|\tilde\sigma_{+})>0$, so that the denominator appearing in the right-hand side of~\eqref{error second} is strictly greater than zero. Also, since $\mathfrak{D}^{\epsilon_1}_{h}(\rho)\ge-\log(1-\epsilon_1)$ independently of $\rho$, the parameter $\epsilon_1$ can be modulated so to compensate, to some extent, eventual lack of resource in the initial state.
\end{remark}

Condition~\eqref{eq:max=min} is stronger than condition~\eqref{error second}, in the sense that if the former is satisfied, the latter is also satisfied. This is because, by multiplying both sides by $2^{D_{\max,\set F}(\tilde{\sigma}\|\tilde\sigma_{+})}-1>0$ (see preceding remark), condition~\eqref{error second} becomes 
	\[
	2^{D_{\max,\set F}(\tilde\sigma\|\tilde\sigma_{+})-\max_{\omega \in \set{F}}{D}_{h}^{\epsilon_1}(\rho\|\omega)}-2^{-\max_{\omega \in \set{F}}{D}_{h}^{\epsilon_1}(\rho\|\omega)}\ge 	2^{D_{\max,\set F}(\tilde\sigma\|\tilde\sigma_{+})-\min_{\omega \in \set{F}}{D}_{h}^{\epsilon_1}(\rho\|\omega)}-1\;,
	\]
and this, if $\max_{\omega \in \set{F}}{D}_{h}^{\epsilon_1}(\rho\|\omega)=\min_{\omega \in \set{F}}{D}_{h}^{\epsilon_1}(\rho\|\omega)$, becomes equivalent to
	\[
	2^{-\max_{\omega \in \set{F}}{D}_{h}^{\epsilon_1}(\rho\|\omega)}\le 1\;,
	\]
which is always trivially satisfied, due to the non-negativity of the hypothesis testing relative entropy. In other words, we have shown the following:

\begin{framed}
\begin{corollary}\label{coro:constant-overlap}
	Given a state $\rho\in\set{D}(\Cm)$, suppose that $\max_{\omega \in \set{F}}D^{\epsilon_1}_{h}(\rho\|\omega)=\min_{\omega \in \set{F}}D^{\epsilon_1}_{h}(\rho\|\omega)$. Then, for any $\sigma$,
	\[
	\mathfrak{D}^{\epsilon_1}_{h}(\rho)	\ge		 \mathfrak{D}^{\epsilon_2}_{\max}(\sigma)\qquad\implies\qquad \rho\succeq_{(\epsilon_1+\epsilon_2)}\sigma\;.
	\]
\end{corollary}
\end{framed}

Corollary~\ref{coro:constant-overlap}, for rank-one $\rho$ and $\epsilon_1=0$, recovers Theorem~2 in Ref.~\cite{Liu-Bu-Takagi-PRL}.

\begin{proof}[Proof of Theorem~\ref{th:3}]
	Case~(i) is easily proved as follows. The condition $\mathfrak{D}^{\epsilon_1}_{h}(\rho)=+\infty$ guarantees the existence of an operator $P\in\mathsf{P}^{\epsilon_1}(\rho)$ such that $\Tr{P\ \omega}=0$ for all $\omega\in\set{F}$. Hence, by constructing a CPTP map as follows:
	\[
	\mE(\cdot):=\Tr{P\ \cdot}\sigma+\Tr{(\openone-P)\ \cdot}\varphi\;,
	\]
	where $\varphi$ is an arbitrarily fixed element of $\set{F}$, we see that $\mE$ maps all free states to $\varphi$, so that $\mE(\set{F})\subseteq\set{F}$, while $\N{\mE(\rho)-\sigma}_1\le 2(1-\Tr{P \rho})\le 2\epsilon_1$.
	
	Case~(ii) follows trivially from the fact that condition $\mathfrak{D}^{\epsilon_2}_{\max}(\sigma)=0$ guarantees the existence of at least one free state which is $\epsilon_2$-close to $\sigma$. Hence, the sought resource morphism is trivially given by the CPTP map that prepares any one such states.
	
	Now we move on to case~(iii). We begin by looking at condition~\eqref{error first}, which is the same in both~(iii.a) and~(iii.b), and rewrite it as follows
	\begin{align}\label{error first rewritten}
	-\log\Tr{P^*\ \omega^*}\ge D_{\max}(\tilde\sigma\|\tilde\sigma_+)\;,
	\end{align}
	where
	\begin{itemize}
		\item the operators $P^*\in\set{P}^{\epsilon_1}(\rho)$ and $\omega^*\in\set{F}$ are chosen to satisfy:
		\begin{align}
		\Tr{P^*\ \omega^*}&=2^{-\mathfrak{D}_h^{\epsilon_1}(\rho) }\\
		&:=\max_{\omega \in \set{F}}\min_{P\in\mathsf{P}^{\epsilon_1}(\rho)}\Tr{P\ \omega}\\
		&=\min_{P\in\mathsf{P}^{\epsilon_1}(\rho)}\max_{\omega \in \set{F}}\Tr{P\ \omega}\label{eq:minimax}\;;
		\end{align}
		the equality in the third line follows from the minimax theorem, for example, in Kakutani's formulation~\cite{kakutani1941,FRENK200446}, whose hypotheses are satisfied since both optimizations range over convex sets and the functional to be optimized is linear, and hence both convex and concave, in its arguments;
		\item the operators $\tilde{\sigma}\in\set{D}(\Cm)$ and $\tilde\sigma_+\in\set{F}$ are chosen so to satisfy:
		\begin{align}
		D_{\max}(\tilde{\sigma}\|\tilde\sigma_+)&=\mathfrak{D}_{\max}^{\epsilon_2}(\sigma)\\
		&:=\min_{\omega \in \set{F}}\min_{\sigma'\in\mathsf{B}^{\epsilon}(\sigma)}D_{\max}(\sigma'\|\omega)\\
		&=\min_{\sigma'\in\mathsf{B}^{\epsilon}(\sigma)}\min_{\omega \in \set{F}}D_{\max}(\sigma'\|\omega)\\
		&=\mathfrak{D}_{\max}(\tilde\sigma)\label{eq:robu-smooth}\;,
		\end{align}
		that is, $\tilde\sigma_+$ achieves the generalized free fraction for $\tilde{\sigma}$ as in Eq.~\eqref{eq:fraction}, namely:
		\begin{align}\label{eq:opt-decomp-sigma-tilde}
		\tilde\sigma_{+}=\nR_{g}(\tilde{\sigma} )\tilde{\sigma} +(1-\nR_{g}(\tilde{\sigma} ))\tilde\sigma_0\;.
		\end{align}
	\end{itemize}
	
	In Ref.~\cite{Buscemi--Sutter--Tomamichel}, condition~(\ref{error first}) alone is shown to be sufficient for the existence of a test-and-prepare CPTP linear map $\mE$ such that $\N{\mE(\rho)-\sigma}_1\le 2(\epsilon_1+\epsilon_2)$ and $\mE(\omega^*)=\tilde\sigma_+$. Such a map is explicitly given as follows:
	\begin{align}\label{eq:the-operation}
	\mathcal{E}(\cdot)=\Tr{P^*\ \cdot}\tilde{\sigma}+\Tr{(\openone-P^{*})\ \cdot}\frac{M\tilde\sigma_{+}-\tilde{\sigma}}{M-1}\;,
	\end{align}
	where, for convenience of notation, we have put $M:=1/\Tr{P^*\ {\omega^*}}=2^{\mathfrak{D}^{\epsilon_1}_{h}(\rho)}$. Without loss of generality, we can assume that $1<M<+\infty$ for the following reasons. First of all, notice that the assumption $\mathfrak{D}^{\epsilon_1}_{h}(\rho)<+\infty$ implies $M< +\infty$. Moreover, we can also assume that $\Tr{P^*\ {\omega^*}}<1$, that is $M>1$, otherwise $\mathfrak{D}^{\epsilon_1}_{h}(\rho)=0$ and, by~\eqref{error first}, $\mathfrak{D}^{\epsilon_2}_{\max}(\sigma)=0$, thus making the situation trivial.
	
	As shown in~\cite{Buscemi--Sutter--Tomamichel}, the above map is CPTP; in order to show that it is a resource morphism, we only need to show that $\mE(\set{F})\subseteq\set{F}$. To this end, let us assume that the input to $\mE$ is an arbitrary $\varphi\in\set{F}$. We need to show that $\mE(\varphi)\in\set{F}$.  By arranging terms, we obtain,
	\begin{align}\label{rewrite}
	\mathcal{E}(\varphi)=\left(1-\frac{1-\Tr{P^{*}\ \varphi}}{1-\Tr{P^{*}\ {\omega^*}}}\right)\tilde{\sigma}+ \left(\frac{1-\Tr{P^{*}\ \varphi}}{1-\Tr{P^{*}\ {\omega^*}}}\right)\tilde\sigma_{+}\;.
	\end{align}
	Again for convenience of notation, let us put $t:=\frac{1-\Tr{P^{*}\ \varphi}}{1-\Tr{P^{*}\ {\omega^*}}}$ and $R:=\nR_{g}(\tilde{\sigma} )$. We now recall the optimal decomposition~(\ref{eq:opt-decomp-sigma-tilde}): by inserting it into~(\ref{rewrite}) and rearranging terms once more, we arrive at
	\begin{align}\label{eq:above-equation}
	\mathcal{E}(\varphi)=(1-t+tR)\tilde{\sigma} +(t-tR)\tilde{\sigma}_0\;.
	\end{align}
	The above relation tells us that $\mathcal{E}(\varphi)$ lies somewhere on the affine line passing through both $\tilde{\sigma}$ and $\tilde{\sigma}_0$. Therefore, in order to have $\mathcal{E}(\varphi)\in\set{F}$, the coefficient $(1-t+tR)$ weighing $\tilde{\sigma}$ 
	must be carefully bounded both from above and from below, so that $\mathcal{E}(\varphi)$ is neither too close to $\tilde{\sigma}$ nor too close to $\tilde{\sigma}_0$, in which case it could end up lying outside $\set{F}$ (see Fig.~\ref{fig:generalrobustness} for a schematic picture).
	
	The upper bound is computed as follows. Since the free fraction is exactly defined as the \textit{maximum} weight of $\tilde{\sigma}$ so that a convex mixture with $\tilde{\sigma}_0$ lies in $\set{F}$, we want to show that the weight of $\tilde{\sigma}$ in~\eqref{eq:above-equation} does not exceed $R$, that is,
	\[
	1-t+tR\le R\;,
	\]
	or, equivalently,
	\begin{equation}\;\label{inproof}
	1-t\le (1-t)R\;.
	\end{equation}
	Since, starting from Eq.~(\ref{eq:minimax}),
	\begin{align*}
	\Tr{P^*\ {\omega^*}}&=\min_{P\in\mathsf{P}^{\epsilon_1}(\rho)}\max_{\omega \in \set{F}}\Tr{P\ \omega}\\
	&=\max_{\omega \in \set{F}}\Tr{P^*\ \omega}\\
	&\ge \Tr{P^*\ \varphi}\ge 0\;,
	\end{align*}
	we see that $t\ge 1$, that is, $1-t\le 0$, and inequality~(\ref{inproof}) automatically holds for any $R\in [0,1]$, without the need to invoke any extra condition.

	Hence, condition~(\ref{error second}) or condition~(\ref{eq:max=min}) are only required to obtain the correct lower bound, that is, to prevent that $\mE(\varphi)$ crosses the border of $\set{F}$ when approaching $\tilde{\sigma}_0$. In order to derive the lower bound, we resort to the construction introduced in Eq.~(\ref{robustness ratio}) and depicted in Fig.~\ref{fig:generalrobustness}. Once a decomposition achieving the generalized free fraction of $\tilde{\sigma}$ is found, $\tilde\sigma_{-}$ is the state on the boundary of $\set{F}$, which is ``antipodal'' with respect to $\tilde\sigma_+$. If we get past it, we end up outside $\set{F}$: we need to make sure this does not happen.

	We begin by assuming that $D_{\max,\set{F}}(\tilde{\sigma}\|\tilde\sigma_{+})<+\infty$, that is, $\tilde\sigma_{+}\neq\tilde\sigma_{-}$. (We recall that $D_{\max,\set{F}}(\tilde{\sigma}\|\tilde\sigma_{+})>0$ is a consequence of the assumption $\mathfrak{D}^{\epsilon_2}_{\max}(\sigma)>0$.) In this case, we need to impose that
	\begin{equation}\label{eq:to-continue}
	1-t+tR \ge  \frac{R\,2^{D_{\max,\set{F}}(\tilde{\sigma}\|\tilde\sigma_{+})}-1}{2^{D_{\max,\set{F}}(\tilde{\sigma}\|\tilde\sigma_{+})}-1}\;.
	\end{equation}
	Before proceeding, we notice that the above inequality, if satisfied, implies in particular $1-t+tR\ge 0$, because $R\,2^{D_{\max,\set{F}}(\tilde{\sigma}\|\tilde\sigma_{+})}=2^{D_{\max,\set{F}}(\tilde{\sigma}\|\tilde\sigma_{+})-D_{\max}(\tilde{\sigma}\|\tilde\sigma_{+})}\ge 1$.
	
	Condition~\eqref{eq:to-continue}, after writing $t$ explicitly again, reads as follows:
	\[
	1-\frac{1-\Tr{P^{*}\ \varphi}}{1-\Tr{P^{*}\ {\omega^*}}}+R\left(\frac{1-\Tr{P^{*}\ \varphi}}{1-\Tr{P^{*}\ {\omega^*}}}\right)\ge \frac{R\,2^{D_{\max,\set{F}}(\tilde{\sigma}\|\tilde\sigma_{+})}-1}{2^{D_{\max,\set{F}}(\tilde{\sigma}\|\tilde\sigma_{+})}-1}\;.
	\]
	Since we are assuming that $\Tr{P^*\ {\omega^*}}<1$,  multiplying both sides by $1-\Tr{P^*\ {\omega^*}}$ does not change the inequality, so we obtain the equivalent condition: 
	\[
	(1-R)\Tr{P^{*}\ \varphi}
	\ge \frac{R\,2^{D_{\max,\set{F}}(\tilde{\sigma}\|\tilde\sigma_{+})}-1}{2^{D_{\max,\set{F}}(\tilde{\sigma}\|\tilde\sigma_{+})}-1}(1-\Tr{P^{*}\ {\omega^*}})+\Tr{P^{*}\ {\omega^*}}-R\;.\]
	After rearranging the right-hand side,  we arrive at
	\[
	(1-R)\Tr{P^{*}\ \varphi}
	\ge  (1-R)\frac{\Tr{P^{*}\ {\omega^*}}2^{D_{\max, F}(\tilde{\sigma}\|\tilde\sigma_{+})}-1}{2^{D_{\max, F}(\tilde{\sigma}\|\tilde\sigma_{+})}-1}\;.
	\]
	Since $R<1$ (because we assumed that $\tilde\sigma\notin\set{F}$, that is, $\mathfrak{D}_{\max}(\tilde\sigma)>0$), we can divide both sides by $(1-R)$ and obtain
	\begin{align}\label{eq:cond-rewritten}
	\Tr{P^{*}\ \varphi}\ge \frac{\frac{1}{M}2^{D_{\max, F}(\tilde{\sigma}\|\tilde\sigma_{+})}-1}{2^{D_{\max, F}(\tilde{\sigma}\|\tilde\sigma_{+})}-1}\;.
	\end{align}
	
	The above condition must be satisfied for any $\varphi\in\set{F}$. Hence, what we really want is a lower bound on $\min_{\varphi \in \set{F}} \Tr{P^{*}\ \varphi}$. Noticing that
	\begin{align}
	\min_{\varphi \in \set{F}} \Tr{P^{*}\ \varphi}&\ge\min_{P\in\mathsf{P}^{\epsilon_1}(\rho)}\min_{\varphi \in \set{F}} \Tr{P\ \varphi}\\
	&=\min_{\omega \in \set{F}}\min_{P\in\mathsf{P}^{\epsilon_1}(\rho)} \Tr{P\ \omega}\\
	&=2^{-\max_{\omega \in \set{F}}D^{\epsilon_1}_{h}(\rho\|\omega)}\;,
	\end{align}
	condition~(\ref{eq:cond-rewritten}) holds whenever the following, stricter condition holds, that is,
	\begin{align*}
	2^{-\max_{\omega \in \set{F}}D^{\epsilon_1}_{h}(\rho\|\omega)}
	&\ge\frac{\frac{1}{M}2^{D_{\max, F}(\tilde{\sigma}\|\tilde\sigma_{+})}-1}{2^{D_{\max, F}(\tilde{\sigma}\|\tilde\sigma_{+})}-1}\\
	&= \frac{2^{D_{\max,\set
				F}(\tilde{\sigma}\|\tilde\sigma_{+})-\min_{\omega \in \set{F}}D^{\epsilon_1}_{h}(\rho\|\omega)}-1}{2^{D_{\max,\set
				F}(\tilde{\sigma}\|\tilde\sigma_{+})}-1} \;.
	\end{align*}
	Hence, condition~(\ref{error second}) guarantees that $\mE(\varphi)\in\set{F}$ for any $\varphi\in\set{F}$, that is, that the operation $\mE$ defined in~\eqref{eq:the-operation} is a valid resource morphism.
	
	Let us finally consider the case in which $D_{\max,\set{F}}(\tilde{\sigma}\|\tilde\sigma_{+})=+\infty$, that is, $\tilde\sigma_{+}=\tilde\sigma_{-}$. In this case, lower and upper bounds have to coincide, so that the map defined in~\eqref{eq:the-operation} is a resource morphism if and only if $1-t+tR=R$. This can only happen if $R=1$ (but this is excluded because $\tilde\sigma\notin\set{F}$) or if $1-t=0$, that is, if $t=1$ independently of the input $\varphi\in\set{F}$. This is guaranteed if the operator $P^*$ in~\eqref{eq:the-operation} has the same trace on all free states, which is exactly the content of~\eqref{eq:max=min}.
\end{proof}

\bigskip A less general, but simpler, statement stemming from Theorem~\ref{th:3} is the following:

\begin{framed}
\begin{corollary}\label{coro:1}
	With the same notations of Theorem~\ref{th:3}, suppose that the following condition holds,
	\begin{align}\label{eq:stronger-cond}
	\mathfrak{D}^{\epsilon_1}_{h}(\rho)\ge D_{\max,\set{F}}(\tilde{\sigma}\|\tilde\sigma_{+})\;.
	\end{align}
	Then, $\rho \succeq_{\epsilon_1+\epsilon_2} \sigma$.
\end{corollary}
\end{framed}

\begin{proof}
	
	Assuming~\eqref{eq:stronger-cond}, if $D_{\max,\set{F}}(\tilde{\sigma}\|\tilde\sigma_{+})=+\infty$, then also $\mathfrak{D}^{\epsilon_1}_{h}(\rho)=+\infty$. In such a case, we know from Theorem~\ref{th:3}, case~(i), that $\rho\succeq_{\epsilon_1}\sigma$, which of course implies also $\rho\succeq_{(\epsilon_1+\epsilon_2)}\sigma$.
	
	On the other hand, if $D_{\max,\set{F}}(\tilde{\sigma}\|\tilde\sigma_{+})=0$, we know that $\tilde\sigma\in\set{F}$, so that, in fact, $\mathfrak{D}^{\epsilon_2}_{\max}(\sigma)=0$. In other words, we are in case~(ii) of Theorem~\ref{th:3}, and again $\rho\succeq_{(\epsilon_1+\epsilon_2)}\sigma$ holds.
	
	We are hence left to consider the case 
	\begin{align}\label{eq:stronger-cond2}
	+\infty>\mathfrak{D}^{\epsilon_1}_{h}(\rho)\ge D_{\max,\set{F}}(\tilde{\sigma}\|\tilde\sigma_{+})>0.
	\end{align}
	We show that condition~\eqref{eq:stronger-cond2} alone implies both conditions~\eqref{error first} and~\eqref{error second} of case~(iii.a) in Theorem~\ref{th:3}.
	
	Since by definition $D_{\max,\set{F}}(\tilde{\sigma}\|\tilde\sigma_{+})\ge D_{\max}(\tilde{\sigma}\|\tilde\sigma_{+})=\mathfrak{D}^{\epsilon_2}_{\max}(\sigma)$, we immediately see that condition~\eqref{eq:stronger-cond2} implies condition~\eqref{error first}. Hence, we only need to show that also condition~\eqref{error second} is implied. In fact, we can show that~\eqref{eq:stronger-cond2} implies a condition that is even stronger than~(\ref{error second}). Such a condition is the following:
	\[
	0\ge \frac{2^{D_{\max,\set
				F}(\tilde{\sigma}\|\tilde\sigma_{+})-\mathfrak{D}^{\epsilon_1}_{h}(\rho)}-1}{2^{D_{\max,\set
				F}(\tilde{\sigma}\|\tilde\sigma_{+})}-1} \;.
	\]
	If the above is satisfied, also~\eqref{error second} is satisfied, and we can conclude that $\rho\succeq_{(\epsilon_1+\epsilon_2)}\sigma$. The above inequality is satisfied because, as a consequence of $D_{\max,\set F}(\tilde{\sigma}\|\tilde\sigma_{+})>0$, the denominator in the right-hand side is strictly positive, so that the above inequality is equivalent to
	\[
	1\ge2^{D_{\max,\set
			F}(\tilde{\sigma}\|\tilde\sigma_{+})-\mathfrak{D}^{\epsilon_1}_{h}(\rho)}\;,
	\]
	which is satisfied if and only if condition~(\ref{eq:stronger-cond2}) is satisfied.
\end{proof}

A merit of Corollary~\ref{coro:1} is to provide a simple compact sufficient condition, free of supplementary \textit{caveat} like condition~(\ref{error second}), which is difficult to interpret operationally. However, the right-hand side of~\eqref{eq:stronger-cond} is not yet a valid resource monotone. The following result fills the gap.

\begin{framed}
\begin{theorem}\label{th:freecase}
	Given $\rho,\sigma\in\set{D}(\Cm)$ and $\epsilon_1, \epsilon_2 \in [0,1]$, if
	\begin{align}\label{eq:error-conditions}
\mathfrak{D}^{\epsilon_1}_{h}(\rho)	\ge		 \mathfrak{D}_{\max,\set{F}}^{\epsilon_2}(\sigma) \;,
	\end{align}
	then $\rho\succeq_{(\epsilon_1+\epsilon_2)} \sigma$.
\end{theorem}
\end{framed}

Theorem~\ref{th:freecase}, when $\epsilon_2 = 0$ and $\sigma$ is rank-one, recovers Theorem~5 in Ref.~\cite{Liu-Bu-Takagi-PRL} (see also Corollary~17 of~\cite{Bartosz-Bu-Takagi-Liu}).

Theorem~\ref{th:3} and Theorem~\ref{th:freecase} are independent of each other. This is because, on the one hand, it is possible that $\mathfrak{D}^{\epsilon_1}_{h}(\rho)	\ge		 \mathfrak{D}_{\max}^{\epsilon_2}(\sigma)$ even though $\mathfrak{D}^{\epsilon_1}_{h}(\rho)	\centernot{\ge}		 \mathfrak{D}_{\max,\set{F}}^{\epsilon_2}(\sigma)$, so that Theorem~\ref{th:freecase} would be inconclusive. On the other hand, it is possible that $\mathfrak{D}^{\epsilon_1}_{h}(\rho)	\ge		 \mathfrak{D}_{\max,\set{F}}^{\epsilon_2}(\sigma)$ even though neither condition~\eqref{error second} nor~\eqref{eq:max=min} hold, so that Theorem~\ref{th:3} would be inconclusive. In other words, Theorem~\ref{th:3} and Theorem~\ref{th:freecase} in general apply to two different regimes and are logically independent of each other. Nonetheless, since $\tilde\sigma\in\set{B}^{\epsilon_2}(\sigma)$ and $\tilde\sigma_+\in\set{F}$, we see that 
$D_{\max,\set{F}}(\tilde{\sigma}\|\tilde\sigma_{+}) \ge\mathfrak{D}_{\max,\set{F}}^{\epsilon_2}(\sigma)$. This implies that Corollary~\ref{coro:1} above can be as well derived as a consequence of Theorem~\ref{th:freecase}. 

\begin{proof}
	We begin by noticing that, if $\mathfrak{D}^{\epsilon_1}_{h}(\rho)=+\infty$, we are back to case~(i) of Theorem~\ref{th:3}. Also, if $\mathfrak{D}_{\max,\set{F}}^{\epsilon_2}(\sigma)=0$, then also $\mathfrak{D}_{\max}^{\epsilon_2}(\sigma)=0$, and we are back to case~(ii) of Theorem~\ref{th:3}. In what follows we will hence assume that $+\infty>\mathfrak{D}^{\epsilon_1}_{h}(\rho)\ge\mathfrak{D}_{\max,\set{F}}^{\epsilon_2}(\sigma)>0$.
		
	Let us define $P^*,{\omega^*},\tilde{\sigma},\tilde{\sigma}_{+}$ as the optimizers achieving the quantities that appear in condition~(\ref{eq:error-conditions}), that is,
	\begin{align}
	&\mathfrak{D}^{\epsilon_1}_{h}(\rho)	:=\min_{\omega \in \set{F}}D_{h}^{\epsilon_1}(\rho\|\omega)=D_{h}^{\epsilon_1}(\rho\|{\omega^*})=-\log\Tr{P^*\ {\omega^*}}\\\label{opt robustness}
	&\mathfrak{D}_{\max,\set{F}}^{\epsilon_2}(\sigma):=\min_{\omega \in  \set{F}}D_{\max,\set{F}}^{\epsilon_2}(\sigma\|\omega)=D_{\max,\set{F}}(\tilde{\sigma}\|\tilde{\sigma}_{+})\;.
	\end{align}
	Notice that while $\tilde{\sigma},\tilde{\sigma}_{+}$ were used in Theorem~\ref{th:3} to denote the optimizers achieving $\mathfrak{D}_{\max}^{\epsilon_2}(\sigma)$, for the sake of this proof the same symbols are used to denote the optimizers achieving $\mathfrak{D}_{\max,\set{F}}^{\epsilon_2}(\sigma)$.
	
	Writing $M:=1/\Tr{P^*\ {\omega^*}}$, that is,
	\[\frac1M=\Tr{P^*\ {\omega^*}}=\max_{\omega \in \set{F}}\min_{P\in\mathsf{P}^{\epsilon_1}(\rho)}\Tr{P\ \omega}\;,\]
	we define the map
	\begin{align}\label{CPTP map}
	\mathcal{E}(\cdot)=\Tr{P^*\ \cdot}\tilde{\sigma}+(1-\Tr{P^{*}\ \cdot})\frac{M\tilde{\sigma}_{+}-\tilde{\sigma}}{M-1}\;.
	\end{align}
	Notice that, with respect to the map constructed in~(\ref{eq:the-operation}), the above map uses the same operator $P^*$, but prepares different states depending on the outcome. As before, moreover, it is possible to assume without loss of generality that $1<M<+\infty$.
	
	Since $D_{\max,\set{F}}(\tilde{\sigma}\|\tilde{\sigma}_{+})\ge D_{\max}(\tilde{\sigma}\|\tilde{\sigma}_{+})$, condition~(\ref{eq:error-conditions}) implies that, 
	\begin{align}\label{BST cd}
	D_{h}^{\epsilon_1}(\rho\|{\omega^*})\ge D_{\max}(\tilde{\sigma}\|\tilde{\sigma}_{+})\;,
	\end{align}
	A direct consequence of \cite{Buscemi--Sutter--Tomamichel} is that condition~(\ref{BST cd}),  together with the fact that $\tilde{\sigma} \in\mathsf{B}^{\epsilon_2}(\sigma)$, imply that the map $\mE$ defined in~(\ref{CPTP map}) is a valid CPTP map such that $\frac{1}{2}\N{\mE(\rho)-\sigma}_1\le\epsilon_1+\epsilon_2$. In what follows we show that $\mE$ is, in particular, a resource morphism.
	
	Because $\tilde{\sigma}$ and $\tilde{\sigma}_{+}$ have been chosen as the states that optimize $\mathfrak{D}_{\max,\set{F}}^{\epsilon_2}(\sigma)$, we have $D_{\max,\set{F}}(\tilde{\sigma}\|\tilde{\sigma}_{+})=\mathfrak{D}_{\max,\set{F}}(\tilde{\sigma})=-\log{\nR(\tilde{\sigma})}$. Therefore, we obtain the following decomposition of $\sigma_{+}$,
	\begin{align}\label{eq:above}
	\tilde\sigma_{+}=\nR(\tilde{\sigma} )\tilde{\sigma} +(1-\nR(\tilde{\sigma} ))\tilde\sigma_0\;,
	\end{align}
	with $\tilde\sigma_{0}\in\set{F}$. By plugging~\eqref{eq:above} in~(\ref{CPTP map}), and considering as input to the map an arbitrary free state $\varphi\in\set{F}$, we reach the following
	\begin{align}\label{eq:above-equation2}
	\mathcal{E}(\varphi)=(1-t+tR)\tilde{\sigma} +(t-tR)\tilde{\sigma}_0\;,
	\end{align}
	where, for the sake of notation, we put $t:=\frac{1-\Tr{P^{*}\ \varphi}}{1-\Tr{P^{*}\ {\omega^*}}}$ and $R:=\nR(\sigma)$. Notice that while the proof of Theorem~\ref{th:3} is obtained by working with the \textit{generalized} free fraction, in this proof we are mostly working with the free fraction.
	
	We need to show that $\mE(\varphi)\in\set{F}$, for all $\varphi\in\set{F}$. To that end, we only need to show that  the weight in front of $\tilde{\sigma}$ in~(\ref{eq:above-equation2}) is non-negative and upper bounded by $R$.
	
	In order to show that it does not exceed $R$, we proceed as follows. In the proof of Theorem~\ref{th:3}, we have shown that $t\ge 1$, so that $1-t+tR\le R$, that is, $R(t-1)\le t-1$, holds automatically for any $R\in[0,1]$.
	
	In order to show the weight of $\tilde{\sigma}$ is non-negative, it suffices to show that
	\[R\ge 1-\frac{1}{t}\;.\]
	Since $t\le \frac{1}{1-\Tr{P^{*}\ {\omega^*}}}=\frac{M}{M-1}$, we have that $1-t^{-1}\le 1-\frac{M-1}{M}$, so that the above is satisfied whenever
	\[R\ge \frac{1}{M}=\Tr{P^*\ {\omega^*}}\;,\]
	that is to say
	\[\nR(\tilde{\sigma} )=2^{-\mathfrak{D}_{\max,\set{F}}^{\epsilon_2}(\sigma)}\ge \Tr{P^*\ {\omega^*}}=2^{-\mathfrak{D}^{\epsilon_1}_{h}(\rho)	}\;.\]
\end{proof}

\section{Applications and examples}\label{sec:applications}

In this section we apply Theorems~\ref{th:3} and~\ref{th:freecase} to some situations of physical interest, and show how we can not only rederive, but sometimes also strengthen, previous results.

\subsection{Singleton Resource Theories}\label{sec:singleton}

We begin this section by considering the special case of singleton resource theories, in which the set of free states $\set{F}$ comprises only one element. This scenario includes the resource theory of athermality, namely, the case in which free operations are those that preserve the thermal state of the system~\cite{Horodecki2013,brandao-et-al-2013-athermality,brandao-et-al-second-laws,buscemi-2015}, which in turns provide the backbone of the resource theory of quantum thermodynamics \cite{Gour-Jennings-Buscemi-Duan}. More generally, when the output singleton is allowed to differ from the input one, this is referred to as the resource theory of asymmetric distinguishability, whose optimal rates have been studied in~\cite{Matsumoto,Buscemi--Sutter--Tomamichel,Wang--Wilde}.

In the singleton case, the log-robustness typically is infinite, and the applicability of Theorem~\ref{th:freecase} is quite limited. On the contrary, Theorem~\ref{th:3} can still be useful, even in the case of a singleton $\set{F}$. Indeed, Theorem~\ref{th:3} reduces in the singleton case to Lemma~3.3 of~\cite{Buscemi--Sutter--Tomamichel}, which is good enough to serve as the starting point to study optimal asymptotic interconversion rates.

\begin{framed}
\begin{proposition}\label{prop:thermal}
	Consider an input system, with initial state $\rho\in\set{D}(\Cm)$ and free singleton $\set{F}=\{\gamma \}$, and an output system, with target state $\sigma\in\set{D}(\Cn)$ and free singleton $\set{F}'=\{\gamma' \}$. If the following condition holds,
	\begin{equation}\label{singleton case}
	D^{\epsilon_1}_{h}(\rho\|\gamma)\ge D^{\epsilon_2}_{\max}(\sigma\|\gamma'),
	\end{equation}
	then $\rho \succeq_{(\epsilon_1+\epsilon_2)} \sigma$.
\end{proposition}
\end{framed}

\begin{proof}
We can restrict ourselves to consider only case~(iii.b) of Theorem~\ref{th:3}, because for a singleton $\set{F}'=\{\gamma' \}$, whenever $\sigma\neq\gamma'$, one has $D_{\max,\set{F'}}(\sigma\|\gamma')=+\infty$. But since also the input free set $\set{F}$ is a singleton, we have
\[
\min_{\omega \in \set{F}}D_h^{\epsilon_1}(\rho\|\omega)=\max_{\omega \in \set{F}}D_h^{\epsilon_1}(\rho\|\omega)\;,
\]
and condition~\eqref{eq:max=min} is automatically satisfied.
\end{proof}

\subsection{Resource Theory of Bipartite Entanglement}\label{sec:entanglement}

Next, we specialize our results to the resource theory of entanglement. We begin by considering bipartite entanglement, namely, the case in which $\set{F}$ is the set of all separable states of a given bipartite system. Resource morphisms are given by separability-preserving (or non-entangling) CPTP maps, usually denoted as SEPP. One-shot entanglement distillation and dilution under SEPP have been studied in~\cite{brandao--datta}. In what follows we show how our Corollary~\ref{coro:1} is able to guarantee the existence of a SEPP transition directly mapping $\rho$ to $\sigma$, even in situations in which the results of Ref.~\cite{brandao--datta} cannot guarantee the existence of a ``distill-and-dilute'' transition.

In order to illustrate the point, it is enough to consider the exact case, that is, $\epsilon_1 = \epsilon_2 = 0$. The same conclusions hold also in the approximate case, however, some care must be taken in that while here we use the trace-distance to measure approximations, Ref.~\cite{brandao--datta} uses the fidelity. Trace-distance and fidelity are well-known to be equivalent \cite{wildebook13,Nielsen--Chuang} , but approximation parameters must be changed: we leave it to the interested reader to work out the exact factors.

By rewriting the main results of~\cite{brandao--datta} using our notation, the zero-error one-shot SEPP-distillable entanglement $E^{(1)}_{D,\text{\text{SEPP}}}(\rho)$ and the zero-error one-shot \text{SEPP}-entanglement cost $E^{(1)}_{C,\text{SEPP}}(\sigma)$ satisfy
\begin{align}
E^{(1)}_{D,\text{SEPP}}(\rho)\ge \lfloor \mathfrak{D}_{\min}(\rho)\rfloor
\end{align}
and
\begin{align}
E^{(1)}_{C,\text{SEPP}}(\sigma)\le \mathfrak{D}_{\max,\set{F}}(\sigma)+1\;,
\end{align}
respectively. These two relations together guarantee that it is possible to exactly go from $\rho$ to $\sigma$ via SEPP (passing through the maximally entangled state) if
\[
\lfloor \mathfrak{D}_{\min}(\rho)\rfloor\ge \mathfrak{D}_{\max,\set{F}}(\sigma)+1\;,
\]
which is more restrictive than what Theorem~\ref{th:freecase} says, that is,
\[
\mathfrak{D}_{\min}(\rho)\ge \mathfrak{D}_{\max,\set{F}}(\sigma)\;.
\]
This is possible because we do not require the transformation to pass through the maximally entangled state, but we allow it to go directly from $\rho$ to $\sigma$.

\begin{remark}
When working within the resource theory of entanglement, especially in the one-shot regime, it is customary to allow the output system to differ from the input one. Consequently, also the set of free states changes from $\set{F}$ to $\set{F}'$. As already noticed, our bounds can be straightforwardly extended to cover this situation as well: in such a case, all quantities related to the input state $\rho$ will be computed with respect to the input free set $\set{F}$, while all quantities related to the output state $\sigma$ will be computed with respect to the output free set $\set{F}'$.
\end{remark}

\subsection{Existence of a Maximally Resourceful State and Weak-Converse Bounds for Distillation and Dilution}\label{sec:maxres}

In this section we show how Corollary~\ref{coro:constant-overlap} and Theorem~\ref{th:freecase} can be used to formulate sufficient conditions that guarantee that an element $\alpha$ is maximally resourceful, in the sense of Definition~\ref{def:max-res-state}. We also address the related problem of deciding when Corollary~\ref{coro:constant-overlap} and Theorem~\ref{th:freecase} are optimal, i.e., when the sufficient conditions they formulate become also necessary. For the sake of the presentation, we focus here on the case of exact transitions, that is, $\epsilon_1=\epsilon_2=0$, keeping in mind, however, that the results of Section~\ref{sec:main} allow us to go beyond the exact case and to speak of, e.g., almost-maximally resourceful elements.

We begin with the following fact (see also Corollary~4 in~\cite{Liu-Bu-Takagi-PRL}):

\begin{framed}
	\begin{proposition}\label{prop:max-resourceful}
		The following statements hold.
		\begin{enumerate}[(i)]
			\item Let $\alpha\in\set{D}(\Cd)$ be such that $\mathfrak{D}_{\min}(\alpha)=\max_{\rho \in\set{D}(\Cd)} \mathfrak{D}_{\max}(\rho)$, and $\Tr{\omega\ \Pi_{\alpha}}=\operatorname{constant}$, for any $\omega \in \set{F}$. Then $\alpha$ is maximally resourceful in $\set{D}(\Cd)$.
			\item Let $\alpha\in\set{D}(\Cd)$ be such that $\mathfrak{D}_{\min}(\alpha)=\max_{\rho \in\set{D}(\Cd)} \mathfrak{D}_{\max,\set{F}}(\rho)$. Then $\alpha$ is maximally resourceful in $\set{D}(\Cd)$.
		\end{enumerate}
	\end{proposition}
\end{framed}

\begin{proof}
	
	Case~(i): being $\Tr{\omega\ \Pi_\alpha}$ constant for any $\omega \in \set{F}$, the assumptions in Corollary~\ref{coro:constant-overlap} are satisfied with $\epsilon_1=\epsilon_2=0$. The proof then follows trivially, from the assumption that $\mathfrak{D}_{\min}(\alpha)=\max_{\rho \in\set{D}(\Cd)} \mathfrak{D}_{\max}(\rho)\ge \mathfrak{D}_{\max}(\sigma)$ for any $\sigma \in \set{D}(\Cd)$.
	
	Case~(ii): in this case we apply Theorem~\ref{th:freecase}, and again the proof follows trivially, from the assumption that $\mathfrak{D}_{\min}(\alpha)=\max_{\rho \in\set{D}(\Cd)} \mathfrak{D}_{\max,\set{F}}(\rho)\ge \mathfrak{D}_{\max,\set{F}}(\sigma)$.
	
\end{proof}

\begin{remark}
	Since, for any $\rho,\sigma$, $D_{\min}(\rho\|\sigma)\le D_{\max}(\rho\|\sigma)\le D_{\max,\set{F}}(\rho\|\sigma)$, condition~(i) in Proposition~\ref{prop:max-resourceful} above implies that $\mathfrak{D}_{\min}(\alpha)=\mathfrak{D}_{\max}(\alpha)=\max_{\rho \in\set{D}(\Cd)} \mathfrak{D}_{\min}(\rho)=\max_{\rho \in\set{D}(\Cd)} \mathfrak{D}_{\max}(\rho)$; analogously, condition~(ii) implies $\mathfrak{D}_{\min}(\alpha)=\mathfrak{D}_{\max,\set{F}}(\alpha)=\max_{\rho \in\set{D}(\Cd)} \mathfrak{D}_{\min}(\rho)=\max_{\rho \in\set{D}(\Cd)} \mathfrak{D}_{\max,\set{F}}(\rho)$.
\end{remark}

The two sufficient conditions considered in Proposition~\ref{prop:max-resourceful} are independent. For example, both in the resource theory of bipartite entanglement and in the resource theory of coherence a golden state exists, namely, the maximally entangled state and the maximally coherent state, respectively. It is also known that these are both in fact maximally resourceful in their respective theories. However, while the maximally coherent state satisfies condition~(i) but not condition~(ii), the maximally entangled state satisfies condition~(ii) but not~(i): see~\cite{Bartosz-Bu-Takagi-Liu} for the explicit calculation.

Another example is provided by the resource theory of genuine multipartite entanglement, in which the free set is taken to be the set of all biseparable states and resource morphisms correspondingly are defined as biseparability-preserving maps. In this case, it is possible to show by explicit calculation~\cite{Contreras-Tejada-Palazuelos-Vicente,Bartosz-Bu-Takagi-Liu} that the generalized GHZ state, that is,
\[
|\Psi_{\mathsf{GHZ}}^{(N,d)}\>:=\frac{1}{\sqrt{d}}\sum_{i=1}^d|i\>^{\otimes{N}}\;,
\]
satisfies condition~(ii) of Proposition~\ref{prop:max-resourceful}. We conclude, therefore, that $|\Psi_{\mathsf{GHZ}}^{(N,d)}\>$ is maximally resourceful.

The following propositions provide sufficient conditions so that the bounds in Corollary~\ref{coro:constant-overlap} and Theorem~\ref{th:freecase} are optimal. In the following proposition, we make it explicit that the input system (with state space $\set{D}(\Cm)$ and free set $\set{F}$) in general may differ from the output system (with state space $\set{D}(\Cn)$ and free set $\set{F}'$). A related result is Theorem~2 of Ref.~\cite{Liu-Bu-Takagi-PRL}.

\begin{framed}
\begin{proposition}[Weak-converse bounds for dilution]\label{prop:equivalent-dilution}
	When dealing with transitions from an input system $(\Cm,\set{F})$ to an output system $(\Cn,\set{F}')$, the following statements hold.
	\begin{enumerate}[(i)]
		\item Suppose that $\alpha\in\set{D}(\Cm)$ satisfies $\mathfrak{D}_{\min}(\alpha)=\mathfrak{D}_{\max}(\alpha)$; then, for any $\sigma\in\set{D}(\Cn)$ 
		\begin{align}\label{eq:case1}
		\alpha\succeq\sigma\qquad\implies\qquad \mathfrak{D}_{\min}(\alpha)\ge\mathfrak{D}_{\max}(\sigma)\;.
		\end{align}
		\item Suppose that $\alpha\in\set{D}(\Cm)$ satisfies $\mathfrak{D}_{\min}(\alpha)=\mathfrak{D}_{\max,\set{F}}(\alpha)$; then, for any $\sigma\in\set{D}(\Cn)$ 
		\begin{align}\label{eq:case2}
		\alpha\succeq\sigma\qquad\implies\qquad \mathfrak{D}_{\min}(\alpha)\ge\mathfrak{D}_{\max,\set{F}'}(\sigma)\;.
		\end{align}
	\end{enumerate}
\end{proposition}
\end{framed}

\begin{proof}
	Case~(i): suppose that $\alpha\succeq\sigma$, so that there exists a resource morphism $\mE:\set{D}(\Cm)\to\set{D}(\Cn)$ such that $\mE(\alpha)=\sigma$; then,
	\begin{align*}
	\mathfrak{D}_{\min}(\alpha)&=\mathfrak{D}_{\max}(\alpha)\\
	&\ge \mathfrak{D}_{\max}(\mE(\alpha))\\
	&=\mathfrak{D}_{\max}(\sigma)\;,
	\end{align*}
	where the inequality in the second line comes from the fact that $\mathfrak{D}_{\max}$ is a resource monotone.
	
	Case~(ii): suppose that $\alpha\succeq\sigma$, then
	\begin{align*}
	\mathfrak{D}_{\min}(\alpha)&=\mathfrak{D}_{\max,\set{F}}(\alpha)\\
	&\ge \mathfrak{D}_{\max,\set{F}'}(\mE(\alpha))\\
	&=\mathfrak{D}_{\max,\set{F}'}(\sigma)\;,
	\end{align*}
	where the inequality in the second line comes from the fact that $\mathfrak{D}_{\max,\set{F}}$ is a resource monotone.
\end{proof}

An analogous weak converse for distillation is the following (see also Theorem~5 of~\cite{Liu-Bu-Takagi-PRL} for a related result).

\begin{framed}
\begin{proposition}[Weak-converse bound for distillation]\label{prop:equivalent-distillation}
Consider an input system $(\Cm,\set{F})$ and an output system $(\Cn,\set{F}')$, and let $\alpha\in\set{D}(\Cn)$ be a target state such that $\mathfrak{D}_{\max,\set{F}'}(\alpha)=\mathfrak{D}_{\min}(\alpha)$. Then, for any $\rho\in\set{D}(\Cm)$,
\begin{align}
\rho\succeq\alpha\qquad\implies\qquad \mathfrak{D}_{\min}(\rho)\ge \mathfrak{D}_{\max,\set{F}'}(\alpha)\;.
\end{align}
\end{proposition}
\end{framed}

\begin{proof}
If $\rho\succeq\alpha$ then
$\mathfrak{D}_{\min}(\rho)\ge \mathfrak{D}_{\min}(\mE(\rho))=\mathfrak{D}_{\min}(\alpha)=\mathfrak{D}_{\max,\set{F}'}(\alpha)$.
\end{proof}

\begin{remark}
	By looking at the proofs of Theorem~\ref{th:3} and Theorem~\ref{th:freecase}, we see that the resource morphisms used there have been constructed as test-and-prepare quantum channels. As a consequence, Propositions~\ref{prop:equivalent-dilution} and~\ref{prop:equivalent-distillation} above can be interpreted as giving sufficient conditions for which test-and-prepare channels are provably optimal in resource manipulation, despite constituting a very special class among all CPTP maps. 
\end{remark}

A natural question to ask, at this point, is whether density matrices always exist, for which Propositions~\ref{prop:equivalent-dilution} and~\ref{prop:equivalent-distillation} hold, namely, for which test-and-prepare channels provide the optimal resource morphisms. As it turns out, perhaps surprisingly, in any resource theory with non-empty closed and convex $\set{F}$, even if a maximally resourceful element may not exist, a \textit{golden state}, namely, a rank-one density matrix $\Psi_{+}\in \set{D}(\Cd)$ such that $\max_{\rho \in\set{D}(\Cd)} \mathfrak{D}_{\min}(\rho)=\mathfrak{D}_{\min}(\Psi_{+})=\mathfrak{D}_{\max}(\Psi_{+})=\max_{\rho \in\set{D}(\Cd)} \mathfrak{D}_{\max}(\rho)$, can always be found~\cite{Liu-Bu-Takagi-PRL,Bartosz-Bu-Takagi-Liu}. However, before concluding that test-and-prepare morphisms are optimal for golden states, one still needs to verify that, either $\Psi_{+}$ satisfies $\Tr{\omega\ \Psi_+}=\text{constant}$ for all free $\omega$, or $\mathfrak{D}_{\max}(\Psi_{+})=\mathfrak{D}_{\max,\set{F}}(\Psi_{+})$ also holds, and both such extra conditions depend on the actual resource theory at hand. The resource theories of coherence and bipartite entanglement again provide two paradigmatic examples in this sense.

\section{Conclusions: resource comparison without a maximally resourceful state}\label{conclusion}

In this work we have derived sufficient (and, in some cases, necessary) conditions for the existence of a CPTP linear map transforming, up to arbitrary accuracy, an input state $\rho$ into a target state $\sigma$, under the additional condition that a convex non-empty subset of states is mapped into itself. Such a framework is particularly suitable to be applied to generalized resource theories, in which the convex subset represents the set of free states in the theory, so that any transformation that maps free states to free states state is itself free, in the sense that it cannot create resources for free. The conditions that we formulated are expressed in terms of entropic monotones, which are computed independently for the input state and the target state. In this way, we can still speak of the resource's worth of any state, taken individually, even if a privileged maximally resourceful state does not exist, so that the tasks of resource distillation and dilution (which typically are used to quantify the resource content) cannot be defined. Aspects that we did not cover in this work are the scaling of the interconversion rates in the case in which a composition rule is given, and the complexity of numerically computing the entropic monotones used to compare resources.

\acknowledgments

The authors are very grateful to Kaifeng Bu, Zi-Wen Liu, Bartosz Regula, and Ryuji Takagi, for insightful discussions about general resource theories, clarifications about the relations between this work and theirs, and a careful reading of a preliminary version of this work. W.Z. thanks Masahito Hayashi for insightful discussions and hospitality during his visit to the Peng Cheng Laboratory, Shenzhen, China. F.B. acknowledges support from the Japan Society for the Promotion of Science (JSPS) KAKENHI, Grants No.19H04066 and No.20K03746.

\bibliographystyle{alphaurl}
\bibliography{biblio-new}	

\end{document}